\documentclass[draftclsnofoot,onecolumn, 12pt]{IEEEtran}
\usepackage{cite}
\usepackage{color}
\usepackage{stfloats}
\usepackage{amssymb}
\usepackage{amsmath}
\usepackage{amsthm}
\usepackage{algorithmic}
\usepackage{graphicx}
\usepackage{textcomp}
\usepackage{color}
\usepackage{amssymb}
\usepackage{amsmath}
\usepackage{amsthm}
\usepackage{url}
\usepackage{caption}
\usepackage{subcaption}
\usepackage{hyperref}
\usepackage{geometry}
\usepackage{setspace}
\usepackage{authblk}
\usepackage{csquotes}
\usepackage{etoolbox}
\newcommand{\bb}[1]{\mathbf{#1}}

\newcommand{\mm}[1]{\mathrm{#1}}

\newcommand{\Rmnum}[1]{\expandafter\@slowromancap\romannumeral #1@}

\newtheorem{theorem}{Theorem}
\newtheorem{remark}{Remark}
\newtheorem{lemma}{Lemma}

\newtheorem{proposition}{Proposition}

\newtheorem{cor}{Corollary}
\usepackage{etoolbox}

\begin{document}

\title{\huge Performance Comparison Between A Simple Full-Duplex Multi-Antenna  Relay And A Passive Reflecting Intelligent Surface}
\author{Armin Bazrafkan, Marija Poposka, Zoran Hadzi-Velkov, Petar Popovski, and  Nikola Zlatanov
\thanks{A. Bazrafkan and N. Zlatanov were with the Department of Electrical and Computer Systems Engineering, Monash University, Melbourne, VIC 3800, Australia (e-mails: armin.bazrafkan@monash.edu and nikola.zlatanov@monash.edu)}
\thanks{M.  Poposka and Z. Hadzi Velkov are with the Department of Electrical Engineering and Information Technologies, Ss. Cyril and Methodius University, 1000
Skopje, Macedonia. (e-mails: poposkam@feit.ukim.edu.mk and zoranhv@feit.ukim.edu.mk)}
\thanks{P. Popovski is with the Department of Electronic Systems, Aalborg University, Denmark.   (e-mail: petarp@es.aau.dk)}
}

\maketitle

\begin{abstract}
In this paper, we propose to investigate a  single RF chain multi-antenna full-duplex (FD) relay built with $b$-bit analog phase shifters and passive self-interference cancellation. 
Next, assuming only passive self-interference cancellation at the FD relay, we derive the achievable data rate of a system comprised of  a source, the proposed FD relay, and a destination. We then compare the achievable data rate of the proposed FD relaying system with the achievable data rate of the same system but with the FD relay   replaced by an ideal passive RIS.
Our results show that the proposed  relaying system with 2-bit quantized analog phase  shifters  significantly outperforms the  RIS-assisted  system.  In fact, the performance  gains are so large, at least for small to intermediate numbers of antenna elements, that we believe it makes this result of interest to the wireless community. 

The proposed FD relay can also be built with reconfigurable holographic surfaces, one surface for the transmit-side and one for the receive-side. For such a scenario, we derive the energy efficiency of the relay-assisted system and compare it with the RIS-assisted system.  Our numerical results show that the energy efficiency of the relay-assisted system built with reconfigurable holographic surfaces is significantly higher than the energy efficiency of the RIS-assisted system.

Intuitively, the RIS system  is at a disadvantage since there the total transmit power $P_T$ is used entirely by the source, whereas  in the FD relaying system the total transmit power $P_T$ is shared by the source and the FD relay in addition to the noise-cleansing process performed by the decode-and-forwarding at the FD relay.
\end{abstract}


\section{Introduction}
\IEEEPARstart{T}{o}  meet the ever-growing demands for wider bandwidths, the emerging wireless communication standards need to  work in  higher frequency  bands, such as the millimeter (mmWave, 30-100 GHz) and the sub-millimetre (above 100 GHz) bands. Wireless transmission in these bands typically necessitates either a direct line-of-sight (LoS) between the transmitter and the receiver, or an intermediary device with a LoS to both the transmitter and the receiver. The intermediary device may be a conventional relay or a recent alternative known as an  reflecting intelligent surface (RIS) \cite{bib1, 9490259}. An RIS is an electronic surface comprised of reflecting meta elements,  where each element can    reflect the incoming electromagnetic wave and shift its phase such that the overall reflected signal from the RIS is beamformed towards the  desired direction. In that sense, the RIS resembles a full-duplex (FD) amplify-and-forward  relay with a large planar antenna array. 

The recent influx of research papers on RIS-aided communications is attributed to the perceived advantages of this technology over conventional relaying, see \cite{bib3} for example. Although the perception of an RIS outperforming conventional relaying is valid for half-duplex (HD) relaying, as shown in  \cite{bib16}, it is yet unclear whether it is also valid for a practical multi-antenna FD relay. Certainly, this perception is not valid for an ideal FD relay that exhibits zero self-interference and has the same number of antenna elements as the RIS. However, an ideal FD relay is impractical to build in the real world since it would require the same number of RF chains as the number of antenna elements and   sophisticated self-interference cancellation hardware. The natural question in this case is: whether   an RIS can outperform a practical FD relay with the same number of antenna elements as the RIS but built only with one RF chain on the receive side and one RF chain on the transmit side? The aim of this paper is to answer this question. To the best of authors' knowledge, the performance of a single RF chain multi-antenna FD relay has not been compared to the performance of a passive RIS-assisted system yet, and therefore it is not known which of the two systems has a better performance.

Multi-antenna array is the main technology that can significantly increase the spectral efficiency  of mmWave communication systems \cite{6798744, 6736761}. One of the most practically-plausible solutions for building   multi-antenna arrays that exhibit  low power consumption and low complexity  is by adopting analog beamforming \cite{6736750, 8030501, 7389996, 8371237}. Analog beamforming can be  implemented by equipping each antenna element with an analog phase shifter, which can shift the phase of the transmitted/received signal  by  the desired phase, and thereby enable the multi-antenna array to perform transmit/receive beamforming to/from the desired direction in space. In general, there are many different technologies for building phase shifters such as reflective, loaded line, switched delay, Cartesian vector modulator, LO-path phase shifter, and phase-oversampling vector modulator \cite{7370753}. Moreover,  this is a very active field of research, with more-advanced phase shifting technologies being constantly invented. One such technology has been recently proposed in \cite{9318487}, where the authors  propose a low-cost 2-bit phase shifter built using pin-diodes. Having in mind that the reflecting elements at the RIS are also build with pin-diodes \cite{8910627}, it can be concluded that the 2-bit phase shifters at the FD relay have a comparable cost with the reflecting elements at the RIS. Moreover, a FD relay with separated transmit and receive  multi-antenna arrays built with 2-bit phase shifters that only employs passive self-interference suppression has a much lower design complexity and implementation cost as that build with an active self-interference cancellation circuitry. Recently, new antenna technologies have been proposed for analogue beamforming such as  {\it reconfigurable holographic surfaces}, \cite{9110848, 9690474, 9136592, 9826717, 9324910, bib23, deng2021reconfigurable},  reconfigurable reflectarrays, \cite{8023752, 6648436, 8485924}, and reconfigurable transmitarrays \cite{5422701}. Similarly to the passive RIS, these antenna types are based  on programmable metamaterials and therefore the relay would have a comparable size and power consumption as the RIS. 
Finally, relaying with large antenna arrays also has the potential to mitigate the self-interference in FD relays \cite{7929405, 6832435}. All of this motivates the work in this paper.

The main contributions of this paper are as follows. In this paper, we first propose to investigate a simple single RF chain multi-antenna FD relay implemented with $b$-bit analog phase shifters, where  two different antenna arrays are used for transmission and reception, respectively, with passive self-interference suppression between them.  Next,  we derive an achievable data rate when the proposed FD relay is employed to relay the signal between a source and a destination. Finally, we compare the achievable data rate of the proposed FD relaying system with the achievable data rate of the same system but with the FD relay  replaced by an ideal passive RIS. Our results show that the proposed  relaying system with 2-bit quantized analog phase  shifters  significantly outperforms the  RIS-assisted  systems. In fact, the performance  gains are so large, at least for small to intermediate numbers of antenna elements, that we believe it makes this result of interest to the wireless community. 
Next, 
the proposed FD relay can also be built with reconfigurable holographic surfaces, one surface for the transmit-side and one for the receive-side. For such a scenario, we derive the energy efficiency of the relay-assisted system and compare it with the RIS-assisted system.  Our numerical results show that the energy efficiency of the relay-assisted system built with reconfigurable holographic surfaces is significantly higher than the energy efficiency of the RIS-assisted system.

%\textcolor{blue}{The proposed relay has the following features: (1) FD
%operation; (2) massive number of transmit and receive antennas at the relay; (3) analog beamforming/combining with phase quantization; (4) passive self-interference suppression, without any active suppression;
%and (5) only two RF chains, irrespective on the number of transmit and receive antennas, one for the transmit-side and the other for the receive-side. Although
%communications systems with most of these elements have already been studied in the literature, e.g., \cite{6702851, 7756408, 5985554, 8830415, 6810439}, we are the first to prove that a relay with such properties significantly outperforms the corresponding passive RIS.}

We note that we compare the proposed FD relay with a passive RIS and not with an active RIS. The main reason for this is the comparison of the performances between an active RIS and a passive RIS in \cite{9530750}, where it was shown that an active RIS outperforms the passive RIS only for a small number of reflective elements. When an RIS has a medium to a large number of reflective elements, the passive RIS significantly  outperforms the active RIS due to the amplification of noise at the active RIS, please see \cite{9530750}.

This paper is organized as follows. In Sec.~\ref{sec-sm}, we provide the system, relay, and channel models. In Sec.~\ref{sec-ar} , we provide the achievable data rates and energy efficiencies of the relay and RIS systems. In Sec.~\ref{sec-nr}, we provide numerical results and Sec.~\ref{sec-c} concludes the paper.

\section{System, FD Relay, and Channel Models}\label{sec-sm}

In the following, we present the system model, the model of the proposed FD relay, and the channel models. 

\subsection{System Model}
The system model is comprised of a single-antenna source transmitter, $S$, a single-antenna destination receiver, $D$, and an intermediary device that facilitates the transmission between the source and the destination. The intermediary device can either be  a FD relay, as shown in Fig.~\ref{fig_relay}, or an RIS, as shown in Fig.~\ref{fig_irs}. The FD relay and the RIS are assumed to employ the same number of antenna elements. Next, both systems are assumed to use identical total transmit power, denoted by $P_T$. In the case of the RIS system, the total power $P_T$ is used entirely by the source, whereas in the FD relaying system the total power $P_T$ is shared between the source, which uses power $P_S$, and the relay, which uses power $P_R$, such that $P_S+P_R=P_T$ holds.

 We assume that the source and the destination cannot communicate directly due to physical obstacles, whereas the intermediary device has unobstructed LoS to both the source and the destination. As a result, the communication between the source and the destination must be conducted via the intermediary device, the FD relay or the RIS.

\begin{figure}
\centering
\includegraphics[width=1\linewidth]{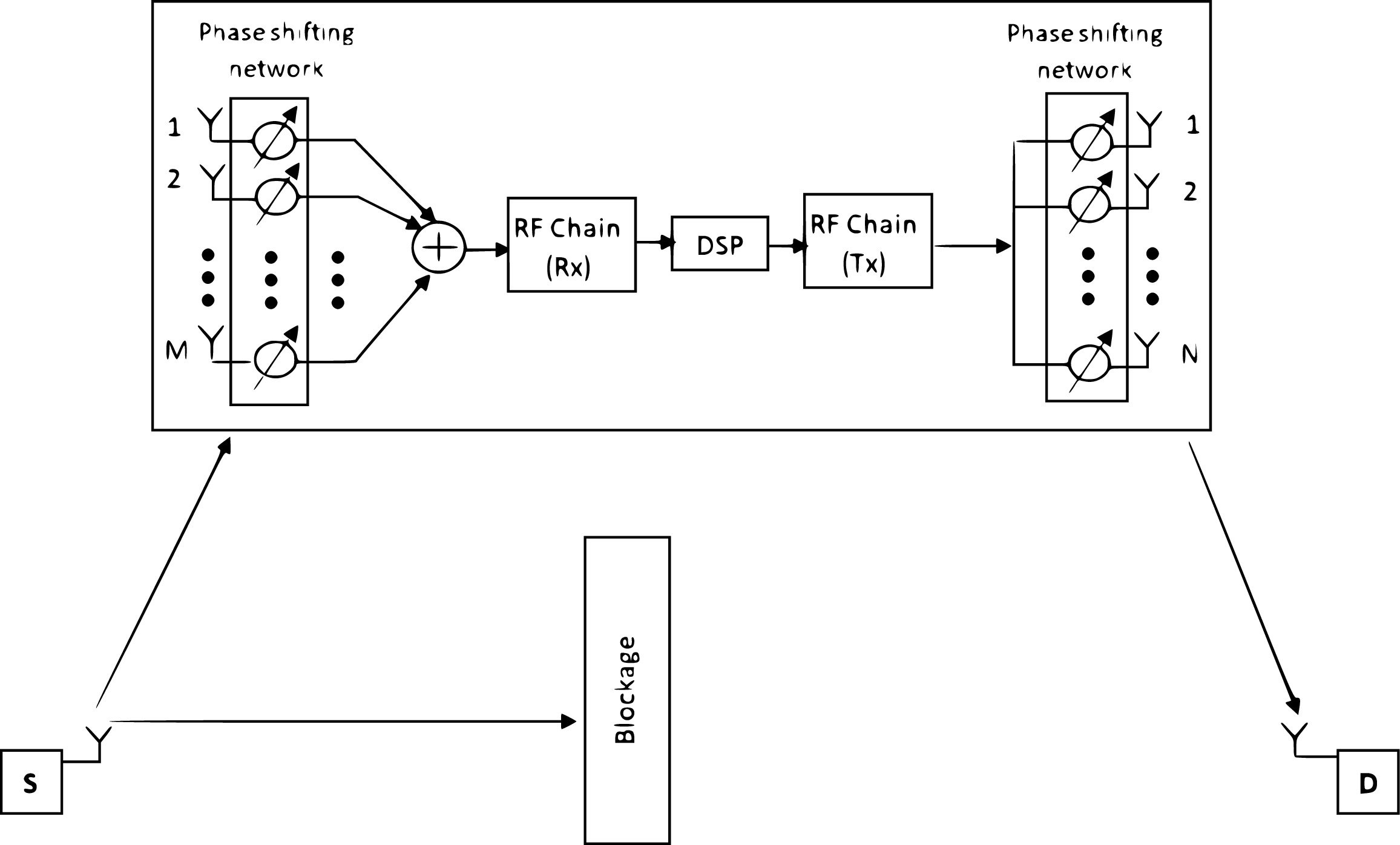} 
\vspace{0mm}
\caption{Model of the full-duplex relay-assisted communications system.} 
\label{fig_relay} 
\end{figure} 

\begin{figure}
\centering
\includegraphics[width=1\linewidth]{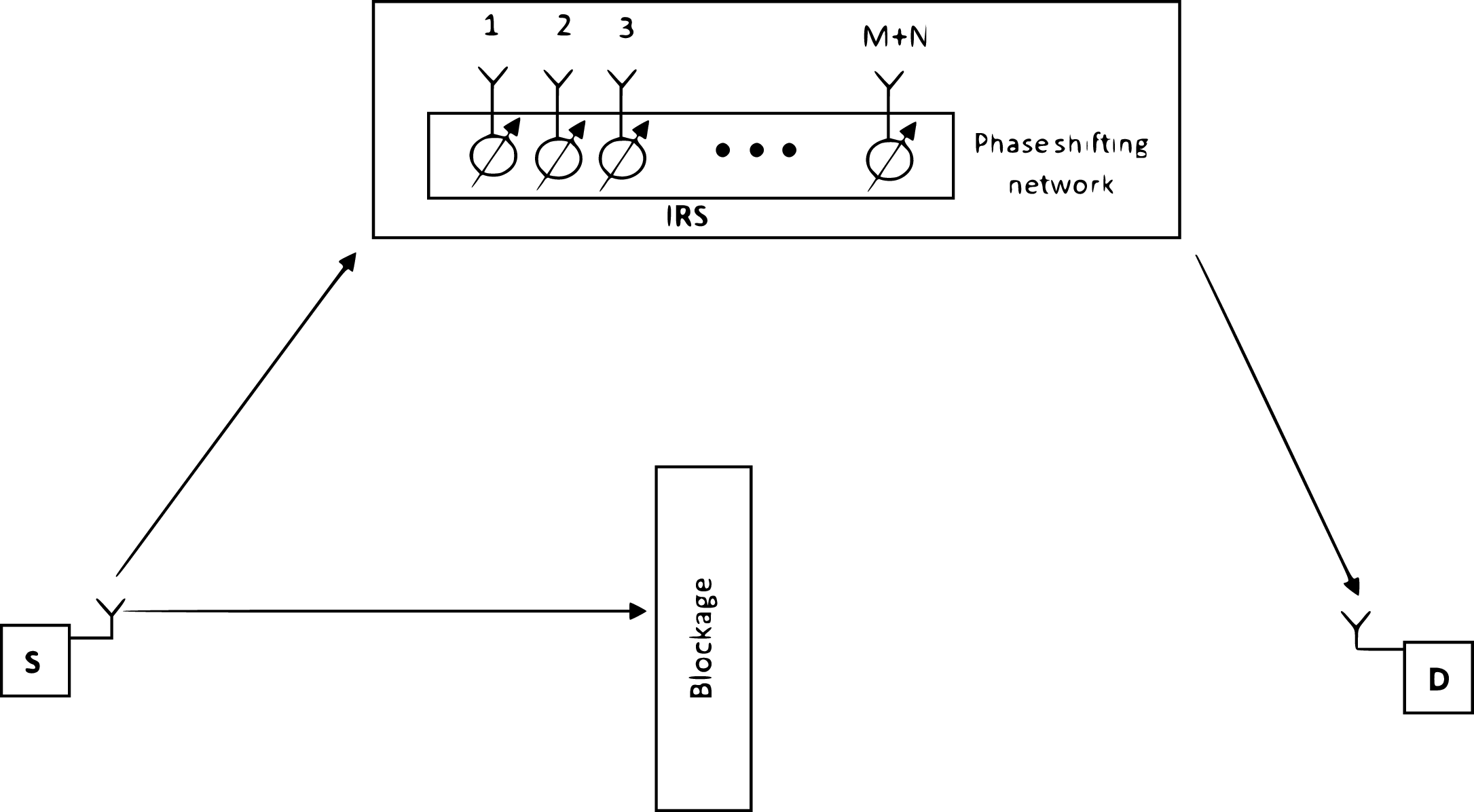} 
\vspace{0mm}
\caption{Model of the RIS-assisted communications system.} 
\label{fig_irs} 
\end{figure}

\subsection{Model of the Proposed FD Relay} 
The relay in the relay-assisted communications system, depicted in Fig.~\ref{fig_relay}, is a multi-antenna FD decode-and-forward relay constructed as follows. The relay is comprised of two mutually isolated planar antenna arrays: ($i$) a receiving planar array with $M$ antennas, and ($ii$) a transmitting planar array with $N$ antennas, where $M+N=K$ and $K$ is the total   number of antenna elements at the FD relay. The receiving and transmitting planner arrays are only passively isolated, hence, there is no active self-interference cancellation  used.

As can be seen  from the structure of the proposed FD relay  illustrated in Fig.~\ref{fig_relay}, the receive-side of the FD relay is comprised of  $M$ antennas, each connected to a $b$-bit analog phase shifter. The $b$-bit analog phase shifter can shift  the phase of the received signal at the given antenna element by a phase   from the following set 
\begin{align}\label{eq_P}
  \mathcal{P}=\Big\{ 0,\;\frac{2\pi}{2^b},\; 2\times\frac{2\pi}{2^b},\; 3\times\frac{2\pi}{2^b},  \cdots,\; (2^b -1)\times \frac{2\pi}{2^b} \Big\},  
\end{align}
where $b=1,2,3...$.
  Next, the phase shifted signals from each receive antenna  are summed via an analog combiner and sent to the receive RF chain. 
  The receive RF chain obtains a baseband digital representation of its input signal, which is then sent to a digital signal processing (DSP) unit for further  processing in the digital domain.
  
On the transmit-side of the proposed FD relay, as illustrated in Fig.~\ref{fig_relay}, a digital transmit signal is sent to a  transmit RF chain, which  produces a passband signal which is then fed into the $N$ transmit antennas.  Each transmit antenna is equipped with  a $b$-bit analog phase shifter which shifts the phase of the transmit signal  by  a phase  from the set $\mathcal P$ in \eqref{eq_P}. The advantage of this transmit-side design is that only one common power amplifier is required for all transmit antennas.

\subsection{Other Methods For Building the Proposed FD Relay}
The design in Fig.~\ref{fig_relay} is a fairly standard  hardware design of   multi-antenna array systems with analog beamforming \cite{6736750, 8030501, 7389996, 8371237}.  However, the large-scale conventional phased antenna arrays may have inherent implementation limitations due to the power consumption and physical size. These limitations can be circumvented by available antenna technologies, such as, {\it reconfigurable holographic surfaces}, \cite{9110848, 9690474, 9136592, 9826717, 9324910, bib23, deng2021reconfigurable},  reconfigurable reflectarrays, \cite{8023752, 6648436, 8485924}, and reconfigurable transmitarrays \cite{5422701}. Similarly to the passive RIS, these antenna types are based  on programmable metamaterials and therefore the relay would have a comparable size and power consumption as the RIS. 

In terms of data rate performances, all of these types of antenna arrays are similar. As a result, we will concentrate on deriving the data rate of the proposed FD relay built with conventional antenna arrays. However, when deriving the energy efficiency of the relay-assisted system,  we will concentrate on the relay built with reconfigurable holographic surfaces, for which power consumption models  are already available in the literature e.g \cite[Sec-V]{9826717}.

\subsection{Channel Models}

Let the channel between the source and the receive-side of the FD relay be represented by the vector $\mathbf{h}_{S} = [h_{S,1}, h_{S,2} \cdots, h_{S,M}]^T$, where $h_{S,m}$ denotes the channel between the source and the relay's $m$-th receive antenna. Let the channel between the transmit-side of the FD relay and the destination  be represented by the vector $\mathbf{h}_{D} = [h_{D,1}, h_{D,2}\cdots, h_{D,N}]^T$, where $h_{D,n}$ denotes the channel between the relay's $n$-th transmit antenna and the destination. The channels $h_{S,m}$ and $h_{D,n}$ are conventionally modelled as  
\begin{align}
h_{S,m} &= \sqrt{\Omega_{S,m}} \, e^{j \phi_{S,m}} \\
h_{D,n}& = \sqrt{\Omega_{D,n}} \, e^{j \phi_{D,n}},
\end{align}
where $\Omega_{S,m}$ and $\phi_{S,m}$ denote the  power gain and phase of the channel between the source and the relay's $m$-th receive antenna, whereas $\Omega_{D,n}$ and $\phi_{D,n}$ denote the  power gain and phase of the channel between the relay's $n$-th transmit antenna and the destination. In mmWave bands, the channel  is  comprised of one very strong LoS component and few very weak non-LoS components. Thereby, each channel can be accurately modeled as static with $\Omega_{S,m}$, $\Omega_{D,n}$, $\phi_{D,n}$ and $\phi_{S,m}$  fixed\footnote{We note that in the case of Rayleigh fading, the channel parameters are not fixed anymore. For completeness, Rayleigh fading will also be investigated in this paper}. We assume that the source and the destination are located in the far-field region of the relay, as defined in \cite{bib19}, and therefore, the corresponding channel power gains are given by \cite{bib19}
\begin{equation}
    \Omega_{S,m} = \Omega_{S}, \forall m, 
\end{equation}
\begin{equation}
    \Omega_{D,n} = \Omega_{D}, \forall n. 
\end{equation}

\subsection{Relay's Self-Interference} 

Due to the FD operation of the proposed relay, there is  self-interference between the receive antenna array and the transmit antenna array. Despite the passive isolation, the relay's receiving antenna array is still exposed to   residual self-interference from its transmit antenna array. In the following, we will derive the maximum amount of possible self-interference at the proposed FD relay.

Let $\mathbf{G} = [g_{mn}]_{M\times N}$ denote the self-interference channel matrix between the relay's transmit array and the relay's receive array, where $g_{mn}$ is the self-interference  channel between the $n$-th transmit antenna and the $m$-th receive antenna of the relay. As shown in \cite{7817779}, the   worst-case scenario with respect to the achievable data rate of the proposed FD relaying system is when $g_{mn}$, $\forall m$ and $\forall n$, are  independent and identically distributed (i.i.d.)   zero-mean Gaussian  random variables. In this paper, we adopt the worst-case scenario in terms of achievable rate for the proposed FD relaying system by assuming that $g_{mn}\sim \mathcal{CN}(0,\sigma_I^2)$, where $\sigma_I^2$ is the  power gain of the self-interference channel between the $n$-th transmit antenna and the $m$-th receive antenna at the relay.  
In the following lemma, we obtain an upper bound on $\sigma_I^2$, i.e., the highest possible value for the power gain of the  self-interference channel.

\begin{lemma}\label{lema_1}
In the absence of passive self-interference cancellation at the proposed FD relay, the power gain of the self-interference channel between any transmit antenna and any receive antenna at the proposed FD relay is upper bounded by 
\begin{align} \label{eqmain}
 \sigma_I^2\leq \frac{1}{M}.
 \end{align}
\end{lemma}
\begin{proof}
Assuming the relay's transmit  power, $P_R$, is equally allocated among the $N$ transmit antennas, the average received power at the relay due to self-interference, denoted by $P_{I}$, is given by 
\begin{align}\label{eq1a}
P_{I}&=\frac{ P_R}{N}E\left\{\left|\sum_{m=1}^{M}\sum_{n=1}^{N}g_{mn}\right|^2\right\}=\frac{ P_R}{N} M N\sigma_I^2\nonumber\\
&= M P_R \sigma_I^2.
\end{align}
Due to the law of conservation of energy,  the average received power at the relay due to self-interference, $P_{I}$, must be smaller or equal to the relay's transmit power, $P_R$, i.e., $P_{I}\leq P_R$ must be satisfied. Combining \eqref{eq1a} with $P_{I}\leq P_R$ yields (\ref{eqmain}). 
\end{proof} 
We assume that   passive insulation is inserted between the relay's transmit and receive antenna arrays with a passive isolation coefficient $\eta \leq 1$. In this case, (\ref{eqmain}) is transformed into $\sigma_I^2 \leq \eta/M$. Note that the typical values of the passive isolation coefficient vary between $10^{-6}$ and $10^{-5}$, c.f. \cite{6702851}. In the rest of the paper, we pessimistically set 
\begin{align}
\sigma_I^2=\frac{\eta}{M},
\end{align}
which leads to the lowest possible achievable data rate of the proposed FD relaying system.

\subsection{Model of the Passive RIS} 

The RIS-assisted communications system is depicted in Fig.~\ref{fig_irs}. The RIS consists of $K$ reflecting elements and has a total area $B$. Following \cite{8936989, bib19, 9569465}, the RIS can be modelled as a planar antenna array, where each antenna element has an area of size $A$, such that $B = KA$ and $A \leq (\lambda/4)^2$. The phase shift of each RIS element can be adjusted such that the incoming signal from  the source is reflected and beamformed towards the destination. In order to facilitate the best possible performance of the RIS, we assume that each phase shifter at the RIS can be set to any value in the range $[0, 2\pi)$, i.e., opposite to the proposed FD relay, there is no phase quantization at the RIS, which maximizes the performance of the RIS. Note that the phase shifters at the RIS   constitute a phase shifting network \cite{bib17}. 

Similar to the relay, we assume that the channel  between $S$ and the $k$-th reflective element at the RIS is denoted by $h_{S,k}$, and the channel between the $k$-th reflective element at the RIS and the destination is denoted by  $h_{D,k}$. The channels $h_{S,k}$   and $h_{D,k}$, for $k=1,2,...,K$, are given by
\begin{align}
h_{S,k} &= \sqrt{\Omega_{S,k}} \, e^{j \phi_{S,k}} \\
h_{D,k}& = \sqrt{\Omega_{D,k}} \, e^{j \phi_{D,k}},
\end{align}
where $\Omega_{S,k}$ and $\phi_{S,k}$ denote the  power gain and phase of the channel between the source and the RIS's $k$-th reflective element, whereas $\Omega_{D,k}$ and $\phi_{D,k}$ denote the  power gain and phase of the channel between the RIS's $k$-th reflective element and the destination. Similar to the relaying system model,  each channel is modeled as static with $\Omega_{S,k}$, $\Omega_{D,k}$, $\phi_{D,k}$, and $\phi_{S,k}$ being  fixed $\forall k$. Moreover, again similar to  the relaying system, we assume that the source and the destination are located in the far-field region of the RIS, as defined in \cite{bib19}, and therefore, the corresponding channel power gains are given by \cite{bib19}
\begin{equation}
    \Omega_{S,k} = \Omega_{S}, \forall k, 
\end{equation}
\begin{equation}
    \Omega_{D,k} = \Omega_{D}, \forall k. 
\end{equation}

\section{Achievable Data Rates And Energy Efficiencies}\label{sec-ar}

In this section, we derive the achievable data rate of the proposed FD relaying system as well as provide the achievable data rate of  the RIS system.

\subsection{Data Rate of the Relay-Assisted System} 
Let $P_S$ denote the transmit power of the source and $P_R$ denote the transmit power of the relay, where $P_S+P_R=P_T$ must hold and $P_T$ is the total available power in the relaying system available for sharing between the source and the relay. Let $x_S$  denote the signal transmitted by the source, such that $E\{|x_S|^2\}=P_S$. Let $\bb {\bar y}_R=[\bar{y}_{R,1}, \bar{y}_{R,2},..., \bar y_{R,M}]^T$ denote the received signal vector at the relay, where $\bar{y}_{R,m}$ is the received signal at the $m$-th receive antenna of the relay before any phase shift is applied. The received vector $\bb {\bar y}_R$ is comprised of the following three components: ($i$) the signal from the source that arrives via the channel $\bb h_S$, given by $x_S\bb h_S$, ($ii$) the noise vector at the relay,  $\bb w_R$, and ($iii$) the self-interference vector, which is found in the following.

Let $x_R$ denotes the information signal of the relay, such that $E\{|x_R|^2\}=P_R$. Let $\bb v=[v_{1},v_{2},...,v_{N}]^T$ denotes the phase shift vector of the relay's transmit antenna array. The relay's $n$-th transmit antenna transmits the signal $x_R  v_n/\sqrt{N}$, where $v_n\in\mathcal{P}$ is the corresponding phase shift applied at this antenna. Thus, the self-interference vector at the relay is given by $x_R \bb G\bb v/\sqrt{N}$, whereas the received vector $\bb {\bar y}_R$ is given by 
\begin{align}\label{e5}
    \bb {\bar y}_R = x_S   \bb h_S+ \frac{1}{\sqrt{N}}x_R   \bb G\bb v  +   \bb w_R.  
\end{align} 
In (\ref{e5}), $\bb w_R =[w_{R,1},w_{R,2},...,w_{R,M}]^T$ denotes the complex additive white Gaussian noise (AWGN) vector at the relay, where $w_{R,m}\sim\mathcal{CN}(0,N_0)$, for $m=1,...,M$, is the complex AWGN at the $m$-th receive antenna of the relay.  

Next, let $\bb u=[u_{1},u_{2},...,u_{M}]^T$  denote the phase shift vector of the relay's receive antenna array, where $u_m\in\mathcal{P}$ is the phase shift at the relay's $m$-th receive antenna. Therefore, the received signal to the relay  is given by
\begin{align} \label{e6}
    y_R &= \bb u^T  \bb {\bar y}_R   \nonumber\\
    &= x_S  \bb u^T \bb h_S+ \frac{1}{\sqrt{N}}x_R  \bb u^T \bb G\bb v +  \bb u^T \bb w_R.
\end{align}

On the other hand, the received signal at the destination is given by 
\begin{align} \label{e7}
    y_D = \frac{1}{\sqrt{N}} x_R \bb v^T\bb h_D  + w_D , 
\end{align} 
where $w_D\sim\mathcal{CN}(0,N_0)$ denotes the complex AWGN at the destination. 

Having derived the input-output relationship of the proposed FD relaying system, we are now ready to derive its achievable data rate, which is provided in the following theorem.

\begin{theorem}
When the phase shift vectors  $\bb u$ and $\bb v$ are given by
\begin{align}
\bb u &=[e^{-j \hat \phi_{S,1}}, e^{-j \hat \phi_{S,1}},..., e^{-j \hat \phi_{S,M}}]^T\label{eq_u}\\
\bb v &=[e^{-j \hat \phi_{D,1}}, e^{-j \hat \phi_{D,1}},..., e^{-j \hat \phi_{D,N}}]^T\label{eq_v} ,
\end{align}
where $\hat \phi_{S,m}$ and $\hat \phi_{D,n}$ are the $b$-bit quantized versions of the channel phases $\phi_{S,m}$ and $\phi_{D,n}$, respectively, and are given by 
\begin{align}
    &\hat \phi_{S,m}=\frac{2m\pi}{2^b}
    \;\mm{if }\;\phi_{S,m}\in\left[\frac{2m\pi}{2^b},\frac{2(m+1)\pi}{2^b}\right),\nonumber\\ 
    &\hat \phi_{D,n}=\frac{2m\pi}{2^b}
    \;\mm{if }\;\phi_{D,n}\in\left[\frac{2m\pi}{2^b},\frac{2(m+1)\pi}{2^b}\right),\nonumber\\
    &\;\mm{for}\;m\in\left\{0,1,\cdots,2^b-1\right\},
\end{align}
   then the achievable data rate of the proposed FD relaying system is given by

\begin{align}\label{eeth1}
R=  &\min \left\{  \log_2\left(1+ \frac{P_S \Omega_S \left(1+(M-1)Q\right)}{N_0+\frac{\eta}{M} P_R} \right) , \right.
\nonumber\\
& \left. \log_2\left(1+\frac{P_R \Omega_D \left(1+(N-1)Q\right)}{N_0}\right)\right\},
\end{align}
where 
\begin{align}\label{eth1.1}
    Q=\left(\frac{2^b}{\pi}\sin\left(\frac{\pi}{2^b}\right)\right)^2.   
\end{align} 
\end{theorem} 
\begin{proof}
Please refer to the Appendix \ref{app1}.
\end{proof}

Obviously, the data rate in \eqref{eeth1} can be maximized by optimizing the values of $M$, $N$, $P_S$, and $P_R$, given the constraints $M+N=K$ and $P_R+P_S=P_T$. Specifically, 
the maximum achievable data rate of the proposed relaying system  is obtained as
\begin{equation}\label{e24}
\begin{aligned}
 \max_{P_R,N} \hspace{2mm} & R\\
  \textrm{s.t.}\hspace{5mm} &  M+N=K\\
  &  P_S+P_R=P_T,
\end{aligned}
\end{equation}    
where the objective function  $R$ is given in \eqref{eeth1} and $P_T$ is the total available power shared between the source and the relay. 

Note that \eqref{e24} can be easily transformed into an unconstrained optimization problem, which  can then be solved numerically by conventional methods for optimization of convex functions, such as the steepest gradient descent method, or   Newton's method. However, this approach does not lead to a closed-form solution. In order to provide a closed-form solution of \eqref{e24}, we provide the following proposition.
 \begin{proposition}\label{cor1}
 We propose the following sub-optimal solution of \eqref{e24}, 
\begin{align}\label{x1.1}
M^*&=\frac{2K}{3},\\\label{x1.3}
N^*&=\frac{K}{3}, \\\label{x1.4}
    P^*_R&=\frac{-\alpha+\sqrt{48P_TKN_0\eta\Omega_{S}\Omega_{D}+\alpha^2}}{6\eta\Omega_{D}},\\\label{x1.2}
    P^*_S&=P_T-P^*_R,
\end{align}
where $ \alpha=2KN_0\Omega_{D}+4KN_0\Omega_{S}. $
 \end{proposition}
\begin{proof}
Please refer to the Appendix \ref{app2}.
\end{proof}

\begin{remark}
  The channel state information (CSI) required for obtaining the phase shift vectors  $\bb u$ and $\bb v$ as per \eqref{eq_u} and \eqref{eq_v} is the $b$-bit quantized information of the phases of the source-relay and relay-destination channels, respectively. This $b$-bit quantized channel phase information can be acquired directly in the analog domain as per \cite{bib20, 8822634}, or by adapting any of the CSI acquisition  methods developed for RIS-assisted networks such as in \cite{9366805}. In that sense, the relay system does not waste more resources in   acquiring the CSI than the corresponding RIS system.   
\end{remark}

\subsection{Rate of the RIS-Assisted System} 

Let $P_T$ denote the transmit power of the source\footnote{Note that the  relaying system also consumes power $P_T$, however, the power $P_T$ is shared between the source and the relay.}. Assuming static LoS channels from the source to the RIS and from the RIS to the destination, the achievable data rate between the source and destination, $ R_{RIS}$, is given by \cite[Eq. (48)]{bib19} 
\begin{equation}\label{eq_r_irs}
R_{RIS}=\log_2\left(1+\frac{K^2 P_T\Omega_{S} \Omega_{D}}{N_0}\right) . 
\end{equation}

\subsection{Rates in Rayleigh Fading Channels} 

The data rates for the relaying system in \eqref{eeth1} and the RIS system in \eqref{eq_r_irs} hold for LoS and static fading. In the following,  for completeness, we provide the data rates for Rayleigh fading channels.

The corresponding expressions for the achievable data rates of the relaying and RIS systems in Rayleigh fading channels are determined by the following corollaries. 

\begin{cor}\label{corrr1}
In Rayleigh fading, the achievable data rate of the proposed relay-assisted communications system, $R^{\textrm Ra}$, is given by
\begin{align}\label{q3.00}
R^{\textrm Ra}=\min\left\{ \log_2\left(1+\gamma_S^{Ra}\right),\log_2\left(1+\gamma_D^{Ra}\right)\right\},
\end{align}
where 
\begin{align}\label{eq3.001}
\gamma_S^{Ra}&=\frac{P_S \Omega_{S} \left( 1  + (M-1)  \frac{\pi}{4} Q \right)}{N_0 + \frac{\eta}{N}P_R} \nonumber\\
\gamma_D^{Ra}&=\frac{P_R\Omega_{D} \left (1 + (N-1)\frac{\pi}{4} Q \right )}{N_0}.
\end{align} 
The above rate in \eqref{q3.00} can be maximized by adopting the sub-optimal solutions in Proposition \ref{cor1}. 

\begin{proof}
Please refer to Appendix \ref{app32}. 
\end{proof}
\end{cor}

In Rayleigh fading, the achievable data rate of the considered RIS-assisted communications system is given by \cite[Eq. (34)]{idk7001}
\begin{align}\label{idk1000}
R_{RIS}^{\textrm Ra} = \log_2 \left(1 + \frac{P_T\Omega_{S}\Omega_{R}K}{N_0}  \times \Big ( 1+(K-1) \Big ( \frac{\pi}{4} \Big )^2 \Big) \right). \notag \\
\end{align}

From \eqref{eq_r_irs} and \eqref{idk1000}, we can see that the Rayleigh fading  reduces the signal-to-noise ratio (SNR) by the factor of $\pi/4$ with respect to the case of static fading.

\subsection{Half-Duplex Relay}

 If instead of the FD mode, the relay works in the HD mode and uses all of its $K$ antennas both for transmission and reception in different time slots, the achievable data rate of the corresponding HD relaying system is given by
\begin{align}\label{e34}
 &R_{HD}=\frac{\log_2\left(1+  \gamma^{HD}_S\right)\log_2\left(1+ \gamma^{HD}_D\right)}{\log_2\left(1+  \gamma^{HD}_S\right)+\log_2\left(1+ \gamma^{HD}_D\right)},
\end{align}
where
\begin{align}
    \gamma^{HD}_S&= \frac{P_T \Omega_{S} \left(1+(K-1)Q\right)}{N_0}\nonumber\\
    \gamma^{HD}_D&=\frac{P_T \Omega_{D} \left( 1+(K-1)Q \right)}{N_0}
\end{align}
Note, the proof is straightforward using the vast available literature on HD relaying, e.g. \cite{reff4}. 

In the numerical results, we also present the achievable rate of a communication system assisted by a HD relay.

\subsection{Energy Efficiency}\label{sec-EE}
In this subsection,  we derive the energy efficiency of the relay-assisted communication system  defined by 
\begin{align} \label{rev2_eq1} 
    EE = \frac{ R}{P_{total}} , 
\end{align} 
where $R$ is the achievable data rate  and $P_{total}$ is the total power consumption of the considered system. 

For completeness, we first provide the  energy efficiency of the   RIS-assisted communications system. The total power consumption of the RIS-assisted system has been derived in \cite{8741198, 9548940, 9257429}
\begin{align} \label{rev2_eq2} 
    P_{total} = \xi P_S + p_c^{(S)} + K p_0^{(RIS)} + p_c^{(RIS)} + p_c^{(D)} , 
\end{align} 
where $\xi$ is the efficiency of the transmit power amplifier, $P_S$ is the RF transmit power of the source (which in the RIS case is $P_S=P_T$), $K$ is the number of the RIS reflecting elements, $p_0^{(RIS)}$ is the power consumption of each RIS element, whereas $p_c^{(S)}$, $p_c^{(RIS)}$, and $p_c^{(D)}$ are the hardware static power consumptions of the source, the RIS, and the destination, respectively. Now, inserting \eqref{rev2_eq2} and $R_{RIS}$ from \eqref{eq_r_irs} instead of $R$ into \eqref{rev2_eq1}, we obtain the energy efficiency of the RIS system.

Analogously  to the derivation of \eqref{rev2_eq2}, the total power consumption of the relaying system is obtained as
\begin{align} \label{rev2_eq3} 
    P_{total}& = \xi P_S + p_c^{(S)} + M p_0^{(relay)} + p_c^{(relay)} + N p_0^{(relay)} \notag \\ 
    &+ \xi P_R + p_c^{(D)} \notag \\ 
    &= \xi \left(P_S + P_R\right) + p_c^{(S)} + \left(M + N \right) p_0^{(relay)} 
    \notag \\ 
    &+ p_c^{(relay)} + p_c^{(D)},  
\end{align} 
where $P_R$ is the RF transmit power of the relay, $p_c^{(relay)}$ is the relay's hardware static power consumption, $p_0^{(relay)}$ is the power consumption of each antenna element at the relay, whereas $M$ and $N$ are the number of receive and transmit antennas. Note that for the relaying system, $P_S + P_R = P_T$ and $K = M + N$. As a result, total power consumption of the relaying system is obtained as
\begin{align} \label{rev2_eq4} 
    P_{total} = \xi P_T + p_c^{(S)} + K p_0^{(relay)} + p_c^{(relay)} + p_c^{(D)}.
\end{align}  

If the HD relay is employed instead of the FD relay, then $p_c^{(S)}$ in (\ref{rev2_eq4}) should be replaced by $p_c^{(S)}/2$ since the source is active only half of the time.

\section{Numerical results}\label{sec-nr}

For a fair comparison between the RIS-assisted system and the relay-assisted system, we adopt the planar array structure considered in \cite{bib19} and thereby assume the area of each transmit/receive antenna element at the relay and each reflecting element at the RIS, denoted by $A$, satisfies $A \leq (\lambda/4)^2$, where $\lambda$ is the signal's carrier wavelength.

 Let $d_{S}$ denote the distances between the source and the RIS/Relay center and $d_D$ denote the distance between the RIS/Relay center and the destination. According to \cite[Eq. (11)]{bib19}, the far-field assumption is justified if the conditions $\sqrt{KA} \leq 3d_S$ and $\sqrt{KA} \leq 3d_D$ are met. In this case, $\Omega_{S,k}$ and $\Omega_{D,k}$ are respectively approximated by \cite[Eqs. (11) and  (31)]{bib19}  as
\begin{equation} \label{eqnew1}
    \Omega_{S,k} \approx \frac{A \cos(\alpha_S)}{4\pi d_S^2} = \Omega_S , \, \forall k,
\end{equation}
\begin{equation} \label{eqnew2}
    \Omega_{D,k} \approx \frac{A \cos(\alpha_D)}{4\pi d_D^2} = \Omega_D , \, \forall k.   
\end{equation}
In (\ref{eqnew1}) and (\ref{eqnew2}), $\alpha_S$ is the angle between the source and RIS/Relay LoS and the RIS/Relay boresight, and $\alpha_D$ is the angle between the destination RIS/Relay LoS and the RIS/Relay boresight.

We set the following system parameters: $\lambda = 0.01$m (i.e., 30 GHz), $A= (\lambda/4)^2$, $d_{S} = d_D= 30$m, $\alpha_{S} = \pi/6$, and $\alpha_{R} = -\pi/6$. The total transmit power is $P_T = 1$W, while the thermal noise power is set to $N_0=10^{-12}$W. We assume that each analog phase shifter at the relay is a 2-bit phase shifter,  which yields $\mathcal{P}=\left\{0,\;\frac{\pi}{2},\; \pi,\;  \frac{3\pi}{2} \right\}$. 

%\footnote{We note that in the far-field regime   the approximate rates given by \eqref{eq_rmax}, \eqref{e34}, and \eqref{eq_r_irs} are extremely accurate even for small $K$  and their accuracy only improves as $K$ grows.} 

\begin{figure}
\centering
\includegraphics[width=1\linewidth]{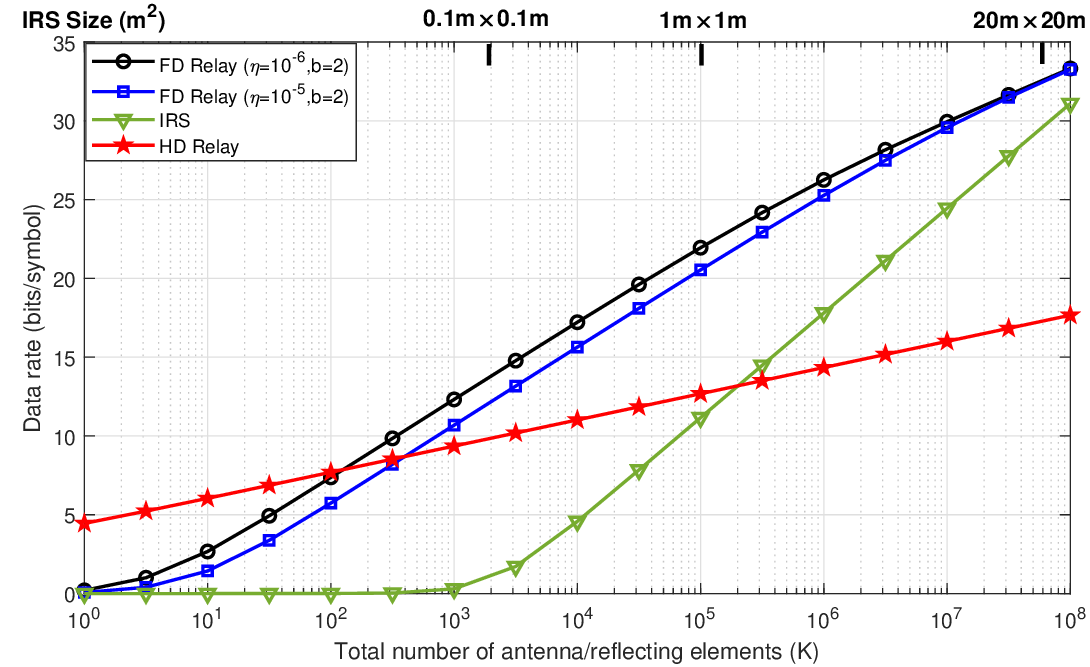} 
\vspace{-4mm}
\caption{Data rate as a function of the number of antenna elements $K$ (LoS).} 
\label{fig1} 
\end{figure}

\begin{figure} 
    \centering
    \includegraphics[width=1\linewidth]{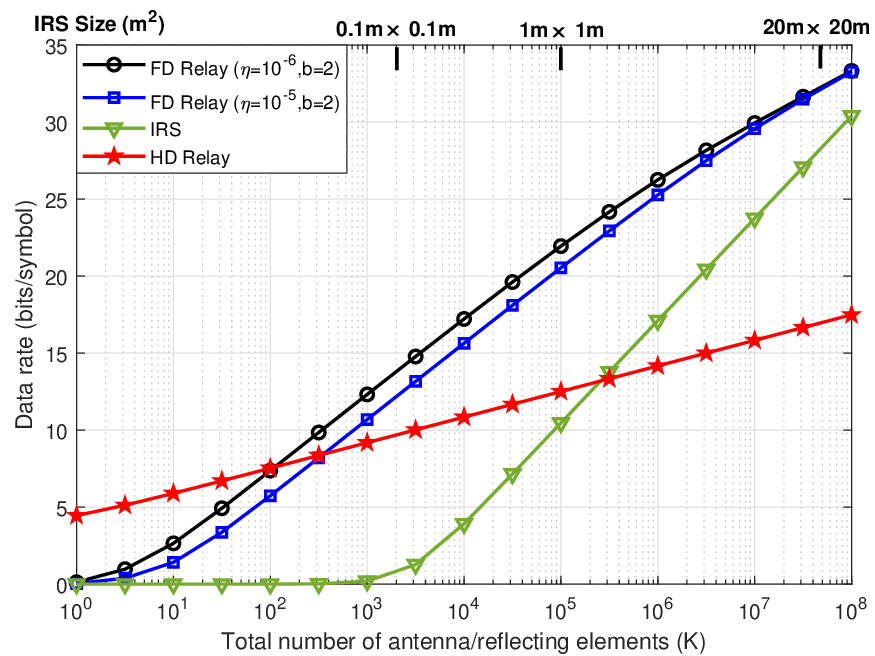}
    \caption{Data rate as a function of the number of antenna elements $K$ (Rayleigh).}
    \label{fig_ray}
\end{figure}

The rates of the three considered systems (RIS, HD Relay, and FD relay) are compared  in Fig.~\ref{fig1} as a function of the total number of RIS/antenna elements $K$ (bottom $x$-axis),  as well as a function of the size of the RIS (top $x$-axis). As can be seen from Fig.~\ref{fig1},    the FD relay significantly outperforms  the  RIS, when the  self-interference  is either $\eta = 10^{-5}$ (i.e., $50$ dB) or $\eta = 10^{-6}$ (i.e., $60$ dB), for any $K<10^8$, i.e., for any RIS size smaller than 20m$\times$20m. Please note that for an RIS of size of 1m$\times$1m, which is an RIS size foreseen to be used in practice, the proposed FD relay achieves almost double the data rate of the RIS.  When $K$ grows to $K=10^8$ or larger, the destination starts to operate in the near field of the relay/RIS, in which case the RIS has a size larger than 20m$\times$20m and therefore it can be considered that the relay/RIS is a part of the source or the destination itself. As a result, in that range,  the RIS does not outperform the FD relay, instead the two have identical performances since  both can be considered as a part of the source or the destination. 

 \begin{figure}
\centering
\includegraphics[width=1\linewidth]{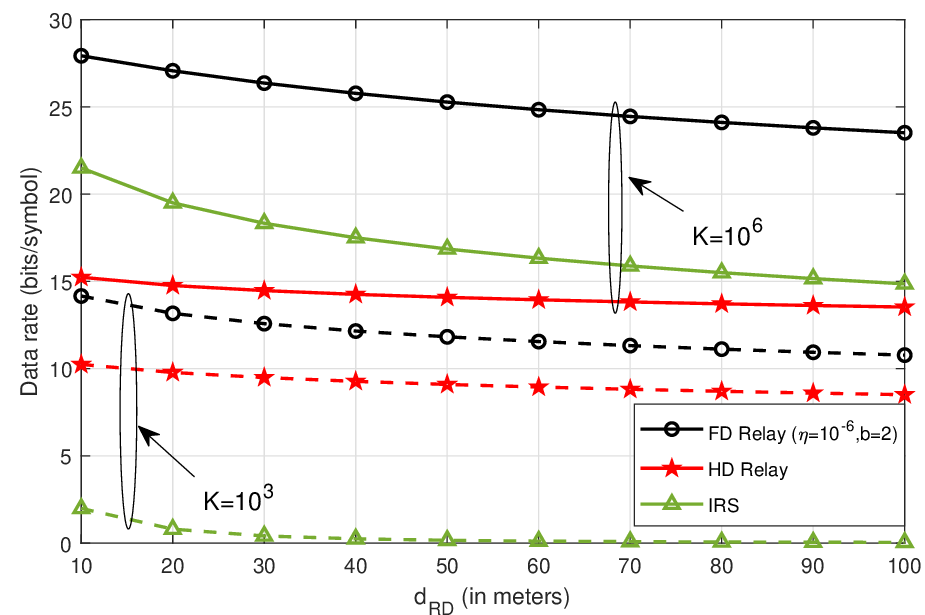} 
\caption{Data rate as a function of the distance $d_{D}$.} 
\label{fig2} 
\vspace{-3mm}
\end{figure}

In Fig.~\ref{fig_ray},  similar to Fig.~\ref{fig1}, the achievable data rates of the proposed FD relay and the RIS are compared under the same assumptions except  instead of a LoS channel a Rayleigh fading channel has been assumed. By comparing Fig.~\ref{fig_ray} and Fig.~\ref{fig1}, we can observe that in the case of Rayleigh fading, the proposed FD relay outperforms the  RIS by a slightly larger margin compared to the case of LoS channel.  This is because by comparing \eqref{eeth1} with \eqref{q3.00} and \eqref{eq_r_irs} with \eqref{idk1000}, it can be seen that  the Rayleigh fading  reduces  the SNR of the relay by a factor of $\pi/4$ and the SNR of the RIS by a factor of $(\pi/4)^2 $. Therefore, in the case of Rayleigh fading channel, the  performance gain of the proposed FD relay compared to the RIS  is even larger.

Fig.~\ref{fig2} compares the data rates of  the three considered systems as a function of   $d_{D}$ when  $d_{S}=25$m  and when the same values as for Fig.~\ref{fig1}  are used for the other parameters. Two sets of curves are presented, one for $K = 10^3$ and the other for $K=10^6$ antenna elements. In both cases, the proposed FD relay significantly outperforms both the RIS and the HD relay.

In Fig.~\ref{fig_1}, we compare the energy efficiencies of the RIS and the relaying systems, given in Sec-\ref{sec-EE}. We assume the relay uses a reconfigurable holographic surface  with $M$ elements at the transmit-side and another reconfigurable holographic surface  with $N$ elements  at the receive-side. For such reconfigurable holographic surface  based antenna implementation at the relay, we consider  that the elements have  power consumption $p_0^{(relay)} = 0.5$mW  based on \cite[Sec-V]{9826717}.
For the RIS, we  consider that the elements have   power consumption  $p_0^{(ris)} = 0.33$mW based on \cite{9548940, 9257429, 9206044}. 
Furthermore, the relay and the RIS are assumed to have equal hardware static power consumptions with $p_c^{(X)} = p_c^{(ris)} = p_c^{(relay)}$, based on \cite{8741198, 9548940, 9257429}. We furthermore set $p_c^{(S)} = p_c^{(D)} = p_c^{(ris)} = p_c^{(relay)} = 1/3$ W, $\xi = 1$ ,  and $P_T = 1$W. 

Fig.~\ref{fig_1}   depicts the energy efficiency versus the number of elements $K$   for the proposed relay built with reconfigurable holographic surfaces and the RIS, respectively. As can be seen from Fig.~\ref{fig_1}, the relay built with reconfigurable holographic surfaces significantly outperforms the RIS in terms of energy efficiency for all values of $K$ shown in the figure.

\begin{figure} 
    \centering
    \includegraphics[width=1\linewidth]{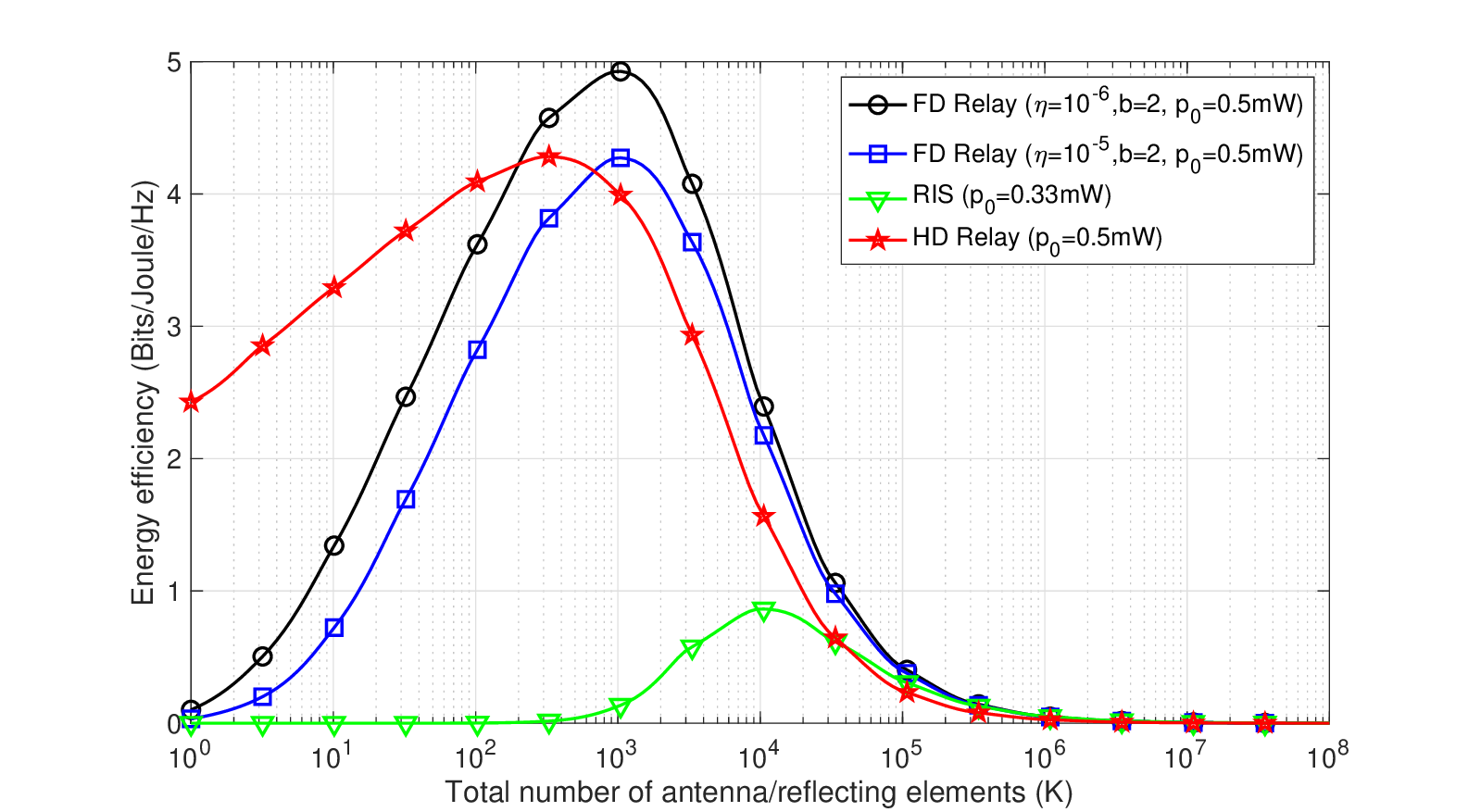} \vspace{-3mm}
    \caption{Energy efficiency as a function of $K$ for Relay  and RIS  (LoS). } \label{fig_1}  
\end{figure}

\section{Conclusion}\label{sec-c}
We proposed to investigate a single RF chain multi-antenna FD relay with discrete analog phase shifters that employs passive self-interference suppression. Next, we  compared the  achievable data  rate   of a system comprised of  a source, the proposed FD relay, and a destination  with the achievable data rate  of the same system but with the FD relay   replaced by an ideal passive RIS. We showed that the FD relaying system with 2-bit quantized analog phase  shifters significantly outperforms the   RIS system. Moreover, we showed that the relay bult with reconfigurable holographic surface significantly  outperforms the RIS in terms of energy efficiency.

In our future work, we plan to pursue a more elaborate
study of the impact of CSI aging as well as a potential design of the proposed FD
relay for multi-user communication and its performance comparison against passive and/or active RIS.

\begin{appendices}

\section{Proof Of Theorem 1}\label{app1}
 
As a result of \eqref{e6},  the signal-to-interference-plus-noise ratio (SINR) at the receive-side of the relay, denoted by  $\gamma_S$,  is given by 
\begin{align}\label{e12}
 \gamma_S=  \frac{E_{x_S}\left\{| x_S  \bb u^T \bb h_S|^2\right\}}{E_{\bb w_R}\left\{\left|\bb u^T \bb w_R \right|^2\right\}+E_{x_R,\bb G}\left\{\left|\sqrt{\frac{1}{N}} x_R  \bb u^T \bb G\bb v\right |^2\right\}}, 
\end{align}
where the subscript of the expectations indicates the random variable(s) with respect to which the expectation is calculated. Now, we have 
$E_{x_S}\left\{| x_S  \bb u^T \bb h_S|^2\right\} = P_S \left| \bb u^T \bb h_S\right|^2$,  which   follows from  $E_{x_S}\left\{|x_S|^2\right\}=P_S$. Next,   when $M$ is large, $ \left| \bb u^T \bb h_S\right|^2$ is given by 
\begin{align}\label{e14}
         &\left| \bb u^T \bb h_S\right|^2 =\left| \sum_{m=1}^{M}\sqrt{\Omega_{S}}e^{j\bar\phi_{S,m}}\right|^2\nonumber\\
    &=M \Omega_{S} + M(M-1) \Omega_{S}  \times\sum_{m=1}^{M}\sum_{\substack{j=1,j\neq m}}^{M}\frac{e^{j\bar\phi_{S,m}}e^{-j\bar\phi_{S,j}}}{M(M-1)}\nonumber\\
    &= M \Omega_{S} + M(M-1) \Omega_{S}  \times E_m\{e^{j\bar\phi_{S,m}}\}E_j\{e^{-j\bar\phi_{S,j}}\}\nonumber\\
    &\overset{(b)}{=} M \Omega_{S}
     + M(M-1) \Omega_{S} \left(\frac{2^b}{\pi}\sin\left(\frac{\pi}{2^b}\right)\right)^2, 
\end{align}
where $\bar \phi_{S,m}=\phi_{S,m}-\hat\phi_{S,m}$, and  $(b)$ follows from the fact that $\bar \phi_{S,m}$ and $\bar\phi_{S,j}$   have uniform distributions over $[-\frac{\pi}{2^b},\frac{\pi}{2^b})$. 
%%This results in 
%\begin{align}\label{e14.1}
% E_m\{e^{j\bar\phi_{SR,m}}\}=E_j\{e^{-j\bar\phi_{SR,j}}\}=\left(\frac{2^b}{\pi}\sin\left(\frac{\pi}{2^b}\right)\right).   
%\end{align}
%Moreover, $ E_{\bb w_R}\left\{\left|\bb u^T \bb w_R \right|^2\right\}$ is given by
%\begin{align}\label{e15}
%  E_{\bb w_R}\left\{\left|\bb u^T \bb w_R \right|^2\right\}= M N_0.
%\end{align}
 Moreover, the average power of the self-interference  is given by 
\begin{align}\label{e16}
      E_{x_R,\bb G}\left\{\left|\sqrt{\frac{1}{N}} x_R  \bb u^T \bb G\bb v\right |^2\right\} \overset{(c)}{=} \frac{P_R}{N}E_{\bb G}\left\{\left| \bb u^T \bb G\bb v\right|^2\right\},
\end{align}
where $(c)$ results from  $E_{x_R}\left\{|x_R|^2\right\}=P_R$. On the other hand, $ E_{\bb G}\left\{\left| \bb u^T \bb G\bb v\right|^2\right\}$ is obtained as 
\begin{align}\label{e17}
 &E_{\bb G}\left\{\left| \bb u^T \bb G\bb v\right|^2\right\} =\sum_{m=1}^{M}\sum_{n=1}^{N}E_{\bb G}\{|g_{mn}|^2\} \nonumber\\
 &+\sum_{m=1}^{M}\sum_{n=1}^{N}\sum_{i\neq m}^{M}\sum_{j\neq n}^{N}E_{\bb G}\{g_{mn}\}E_{\bb G}\{g_{ij}\} \nonumber\\
 &\times e^{-j\hat\phi_{S,m}\hat\phi_{S,i}}e^{-j\hat\phi_{D,n}}e^{j\hat\phi_{D,j}}\nonumber\\
&=\sum_{m=1}^{M}\sum_{n=1}^{N}E_{\bb G}\{|g_{mn}|^2\} 
 =MN\sigma_I^2\overset{(d)}{=} N\eta,
\end{align}
where $(d)$ follows when $ \sigma_I^2\leq \frac{\eta}{M}$ holds with equality.
Inserting \eqref{e17} into \eqref{e16}, we obtain \eqref{e16} as 
\begin{align}\label{e17.1}
  &E_{x_R,\bb G}\left\{\left|\sqrt{\frac{1}{N}} x_R  \bb u^T \bb G\bb v\right |^2\right\}=P_R\eta.   
\end{align}
Now,  considering  \eqref{e17.1},  \eqref{e14}, and the fact that $  E_{\bb w_R}\left\{\left|\bb u^T \bb w_R \right|^2\right\}= M N_0$, we obtain $\gamma_S$ in \eqref{e12} as 
\begin{align}\label{e19}
    \gamma_S =\frac{P_S\left(\Omega_{S} + (M-1) \Omega_{S} \left(\frac{2^b}{\pi}\sin\left(\frac{\pi}{2^b}\right)\right)^2\right)}{N_0+\frac{\eta}{M} P_R}.
\end{align}

On the other hand, from \eqref{e7}, the received SNR at the destination, denoted by $\gamma_D$, is given by 
\begin{align}\label{e20}
        \gamma_D=\frac{E_{x_R}\left\{\left|\frac{1}{\sqrt{N}} x_R \bb v^T\bb h_D \right|^2\right\}}{E_{w_D}\{w_D^2\}} \overset{(e)}{=}\frac{P_R\left|\bb v^T\bb h_D\right|^2}{N N_0},
\end{align}
where $(e)$ results from $E\{|x_R^2|\}=P_R$. In \eqref{e20},  $\left|\bb v^T\bb h_D\right|^2$ can be obtained by replacing $N$ with $M$  and $\Omega_{D}$ with $\Omega_{S}$ in \eqref{e14}. As a result,  $\gamma_D$ is given by 
\begin{align}\label{e22}
  \gamma_D= \frac{P_R\left(  \Omega_{D} +(N-1) \Omega_{D} \left(\frac{2^b}{\pi}\sin\left(\frac{\pi}{2^b}\right)\right)^2\right)}{N_0}.
\end{align}
Finally,  the achievable data rate   is given by 
\begin{align}\label{e22.01}
   R=\textrm{min}\left\{\log_2\left(1+\gamma_S\right),\log_2\left(1+\gamma_D\right)\right\},
\end{align}
where $\gamma_S$ and $\gamma_D$ are given in \eqref{e19} and \eqref{e22}, respectively, which leads to \eqref{eeth1}.
 
 \section{Proof Of Proposition 1}\label{app2}
According to \cite[Th1]{reff4}, in order to maximize the relay rate, the following needs to hold 
\begin{align}\label{ee1}
     \frac{P_S\Omega_{S}\left(1+(M-1)Q\right)}{N_0+\frac{\eta}{M} P_R} =\frac{P_R\Omega_{D}\left(1+(N-1)Q\right)}{N_0}
\end{align}
On the other hand, since $M$ and $N$ are large numbers, we adopt the following highly accurate approximations:  $1+(M-1)Q\approx MQ$ and  $1+(N-1)Q\approx NQ$. Hence, \eqref{ee1} transforms into  
\begin{align}\label{ee2}
  \frac{P_S\Omega_{S}MQ}{N_0+\frac{\eta P_R}{M}}=\frac{P_R\Omega_{D}NQ}{N_0}.  
\end{align}
Denoting $a=\frac{\Omega_{S}}{\Omega_{D}}$, and $\eta_0=\frac{\eta}{N_0}$, we obtain the following equality from \eqref{ee2}
\begin{align}\label{ee3}
    \eta_0NP_R^2+MNP_R-aM^2P_S=0.
\end{align}
Next, in order to maximize the  relay SINR,  we form the following optimization problem 
\begin{align}\label{ee4}
    \max_{P_S,P_R,M,N} \quad & P_RN\nonumber\\
\textrm{s.t.} \quad C_1:\quad& \eta_0NP_R^2+MNP_R-aM^2P_S=0\nonumber\\
  \quad C_2:\quad& P_R+P_S=P_T  \nonumber \\
  \quad C_3:\quad& M+N=K.
\end{align}
We use Lagrangian multipliers method to solve \eqref{ee4}. Hence, the Lagrangian function is formed as 
\begin{align}\label{ee5}
   L&= P_RN+\lambda_1\left(\eta_0NP_R^2+MNP_R-aM^2P_S\right)\nonumber\\
   &+\lambda_2\left(P_R+P_S-P_T \right)+\lambda_3\left( M+N-K\right).
\end{align}
We find the derivative of $L$ with respect to $P_R$, $P_S$, $M$, and $N$ and set them to zero as follows
\begin{align}\label{ee6}
    \frac{\mm dL}{\mm dP_R}&=N+2\lambda_1N\eta_0P_R+\lambda_1MN+\lambda_2=0\\\label{ee6.1}
      \frac{\mm dL}{\mm dN}&=P_R+\lambda_1P_R^2\eta_0+MP_R\lambda_1+\lambda_3=0\\\label{ee6.2}
        \frac{\mm dL}{\mm dM}&=\lambda_1NP_R-2\lambda_1aMP_S+\lambda_3=0\\\label{ee6.3}
         \frac{\mm dL}{\mm dP_S}&=-\lambda_1aM^2+\lambda_2=0
\end{align}
Let $a_{11}=N(2\eta_0P_R+M)$, $a_{21}=P_R^2\eta_0+MP_R$, $a_{31}=NP_R-2aMP_S$, and $a_{41}=-aM^2$. Hence, \eqref{ee6}-\eqref{ee6.3} can be written as  
\begin{align}\label{ee7}
    \begin{bmatrix}
     a_{11}&1&0\\
      a_{21}&0&1\\
        a_{31}&0&1\\
         a_{41}&1&0
    \end{bmatrix}\begin{bmatrix}
     \lambda_1\\
     \lambda_2\\
     \lambda_3
    \end{bmatrix}=\begin{bmatrix}
     -N\\
     -P_R\\
     0\\
     0
    \end{bmatrix}.
\end{align}
The system of equations in \eqref{ee7} consists of four equations and three unknowns. Since we assume that there exists a unique solution, the four equations in \eqref{ee7} must be linearly dependent. Thus, the system determinant must be equal to zero. As a result, the following needs to hold 
\begin{align}\label{ee8}
    \begin{vmatrix}
    a_{11}&1&0&-N\\
      a_{21}&0&1&-P_R\\
        a_{31}&0&1&0\\
         a_{41}&1&0&0
    \end{vmatrix} =(a_{11}-a_{41})P_R-(a_{21}-a_{31})N=0,
\end{align}
which results in 
\begin{align}\label{ee9}
    (a_{11}-a_{41})P_R=(a_{21}-a_{31})N
\end{align}
By substituting values of $a_{11}$, $a_{21}$,$a_{31}$, and $a_{41}$ in \eqref{ee9}, the following equality is formed 
\begin{align}\label{ee10}
    \eta_0NP_R^2+(aM^2+N^2+2aMN)P_R -2aMNP_T=0.
\end{align}
On the other hand, by substituting $P_S=P_T-P_R$  into $C_1$ in \eqref{ee4}, obtained from $C_2$, we have
\begin{align}\label{ee11}
    \eta_0NP_R^2+(aM^2+MN)PR-aM^2P_T=0.
\end{align}
By setting 
\begin{align}
  \alpha&= \eta_0N,\quad b_1=aM^2+N^2+2aMN,\nonumber\\
  c_1&=2aMNP_T,\quad b_2=aM^2+MN,\quad c_2=aM^2P_T,
\end{align}
\eqref{ee10} and \eqref{ee11} can be written as 
\begin{align}\label{ee12}
    \alpha P_R^2+b_1P_R-c_1&=0\\\label{ee12.1}
      \alpha P_R^2+b_2P_R-c_2&=0. 
\end{align}
The roots of \eqref{ee12} and \eqref{ee12.1} are given as 
\begin{align}\label{ee13}
P_{R,1}&=\frac{-b_1+\sqrt{
4\alpha c_1}}{2\alpha} \\\label{ee13.1}  
P_{R,2}&=\frac{-b_2+\sqrt{
4\alpha c_2}}{2\alpha},
\end{align}
respectively. If we assume that $b_1^2\ll 4\alpha c_1$ and $b_2^2\ll 4\alpha c_2$ (due to the existence of $\eta_0$ in $\alpha$), then \eqref{ee13} and \eqref{ee13.1} will be given as follows \begin{align}\label{ee14}
   P_{R,1}&=\frac{\sqrt{4\alpha c_1}}{2\alpha}\\\label{ee14.1}
    P_{R,2}&=\frac{\sqrt{4\alpha c_2}}{2\alpha},
\end{align}
respectively. Since \eqref{ee14} and \eqref{ee14.1} need to be equal, we obtain 
    $c_1=c_2$,
which results in 
\begin{align}\label{ee16}
    M=2N.
\end{align}
Considering $C_3$ in \eqref{ee4}, \eqref{ee16} gives us the optimal $M$, denoted by $M^*$, as 
\begin{align}\label{ee17}
    M^*=\frac{2K}{3}.
\end{align}
Substituting \eqref{ee17} in \eqref{ee14}, we obtain the optimal $P_R$, denoted by $P_R^*$, as 
\begin{align}\label{ee18}
    P_R^*=\sqrt{\frac{4K\Omega_{S}N_0P_T}{3\Omega_{D}\eta}}. 
\end{align}
On the other hand, from $C_2$ and $C_3$, we can obtain the optimal values of $N$ and $P_S$, denoted by $N^*$ and $P_S^*$, as
\begin{align}\label{ee19}
    N^*=K-M^*=\frac{K}{3}\\\label{ee19.1}
    P_S^*=P_T-P_R^*. 
\end{align}
 This approximation is valid when $b_1^2\ll 4\alpha c_1$ and $b_2^2\ll 4\alpha c_2$.
 One issue with the optimal result found for $P_R$, given by \eqref{ee18}, is that when $K$ grows to very large values, $P_R^*$ becomes bigger than $P_T$ which cannot happen. In order to solve this problem, we use $M^*=\frac{2K}{3}$ and $N^*=\frac{K}{3}$ as a starting  point to  maximize the rate. As a result, substituting $M^*=\frac{2K}{3}$ and $N^*=\frac{K}{3}$ in \eqref{ee1} leads to the following condition 
 \begin{align}\label{x2}
      \frac{P_S \Omega_{S} \left(1+(\frac{2K}{3}-1)Q\right)}{N_0+\frac{\eta}{M} P_R}  = \frac{P_R \Omega_{D} \left(1+(\frac{K}{3}-1)Q\right)}{N_0}.
 \end{align}
In \eqref{x2}, we introduce the approximations, $1 + (2K/3-1)Q \approx 2KQ/3$ and 
$1 + (K/3-1)Q \approx KQ/3$, which become tight with increasing $K$ and yields 
 \begin{align}\label{x3}
         \frac{P_S \Omega_{S} \left(\frac{2K}{3}Q\right)}{N_0+\frac{\eta}{M} P_R}  = \frac{P_R \Omega_{D} \left(\frac{K}{3}Q\right)}{N_0}.  
 \end{align}
 Now,  by substituting $P_S$ with $P_T-P_R$, and solving the quadratic equation that results from the equality in \eqref{x3} with respect to $P_R$, we obtain $P_R^*$ and $P_S^*$ given in \eqref{x1.1} and \eqref{x1.2}, respectively. 
 
 \section{Proof of Corollary 1}\label{app32}
Let the fading channel between the source \textit{S} and the $m$-th receive antenna of the relay be denoted by $h_{S,m}$, and is assumed to be a complex Gaussian random variable with zero mean and variance $\Omega_{S,m}$, i.e., $ h_{S,m} \sim \mathcal {CN}(0, \Omega_{S, m}) $. Let the fading channel between the $n$-th transmit antenna of the relay and the destination \textit{D} be denoted by $h_{D,n}$, and is assumed to be a complex Gaussian random variable with zero mean and variance $\Omega_{D,n}$, i.e., $ h_{D,n} \sim \mathcal{CN}(0, \Omega_{D, n}) $. Therefore, the fading channels $h_{S,m}$ and $h_{D,n}$ can be expressed as
\begin{align} \label{Hsr}
h_{S_m}  = \alpha_m \, e^{j \phi_{S,m}}, \qquad h_{D,n}  = \beta_n \, e^{j \phi_{D,n}}, 
\end{align}
where $\alpha_m$ and $\beta_n$ are Rayleigh distributed channel gains between the source and the $m$-th receive antenna at the relay, and between and $n$-th transmit antenna at the relay and the destination, respectively.  Hence, we have \begin{equation}\label{idk1}
    f_{\alpha_m}(x) = \frac{2x}{\Omega_{S,m}} \exp \left(-\frac{x^2}{\Omega_{S,m}} \right), \quad x\geq 0  
\end{equation}
\begin{equation}\label{idk2}
    f_{\beta_n}(x)= \frac{2x}{\Omega_{D,n}} \exp \left(-\frac{x^2}{\Omega_{D,n}} \right), \quad x\geq 0.
\end{equation}
As a result of \eqref{idk1} and \eqref{idk2}, $\mathbb{E}\{\alpha_m^2\} = \Omega_{S_m}, \forall m, \mathbb{E}\{\alpha_m\} = \sqrt{\Omega_{S_m} \pi/4}, \forall m$ and $ \mathbb{E}\{\beta_n^2\} = \Omega_{D_n}, \forall n, \mathbb{E}\{\beta_n\} = \sqrt{\Omega_{D_n} \pi/4}, \forall n $. 
Due to the far field assumption, the average gain of the channel between the source and the $m$-th antenna element of the relay's receive array satisfies  
    $\Omega_{S,m} = \Omega_{S}, \forall m.$
 
Similarly, the average gain of the channel between the $n$-th antenna element of the relay's transmit array and the destination satisfies  
   $ \Omega_{D,n} = \Omega_{D}, \forall n.$

From \eqref{e6}, the SINR at the relay, denoted by $\gamma_S^{Ra}$, can be obtained as 
\begin{align} \label{ez2}
    \gamma_S^{Ra} \overset{(f)}{=}  \frac{P_S\mathbb{E}_{ \mathbf{h}_S}\left\{\left |  \mathbf{u}^T \mathbf{h}_S  \right |^2\right\} }{\mathbb{E}_{ \mathbf{w}_R }\left \{ \left |  \mathbf{u}^T \mathbf{w}_R   \right |^2 \right \} + P_R\mathbb{E}_{\mathbf{G}}\left\{\left |  \sqrt{\frac{1}{N}} \mathbf{u}^T \mathbf{G}\mathbf{v}   \right |^2\right\} },
\end{align}
where $(f)$ results from $E_{x_S}\left\{|x_S|^2\right\}=P_S$ and   $E_{x_R}\left\{|x_R|^2\right\}=P_R$. On the other hand, the only term in \eqref{ez2} different from \eqref{e12} is $\mathbb{E}_{ \mathbf{h}_S}\left\{\left |  \mathbf{u}^T \mathbf{h}_S  \right |^2\right\} $ which can be obtained  as
\begin{align} \label{ez3.001}
    \mathbb{E} \left\{\left |  \mathbf{u}^T \mathbf{h}_{S}  \right |^2\right\} & = \mathbb{E}\left\{ \left |\sum_{m=1}^{M}\alpha_{m} e^{j\phi _{S,m}}e^{-j\widehat{\phi} _{S,m}}\right |^2 \right\}
    \nonumber\\
    &= \mathbb{E}\left\{ \left |\sum_{m=1}^{M}\alpha_{m}e^{j\overline{\phi} _{S,m}}\right |^2 \right\}.
    \end{align}
  Now,   \eqref{ez3.001} 
    can be further simplified as 
    \begin{align}\label{ez3}
   & \mathbb{E} \left\{\left |  \mathbf{u}^T \mathbf{h}_{S}  \right |^2\right\} =  \mathbb{E}\left\{ \sum_{m=1}^{M} \alpha_{m}^2\right\} \nonumber\\
    &+ \mathbb{E}\left\{\sum_{m=1}^{M}\sum_{j=1,j\neq m}^{M}\alpha_{m}\alpha_{j} e^{j\overline{\phi} _{S,m}} e^{-j\overline{\phi} _{S,j}}\right\} =   \sum_{m=1}^{M} \mathbb{E}\left\{\alpha_{m}^2\right\} 
     \nonumber\\
     &+\sum_{m=1}^{M}\sum_{j=1,j\neq m}^{M} \mathbb{E}\left\{\alpha_{m}\right\} \times\mathbb{E}\left\{\alpha_{j}\right\}  \mathbb{E}\left\{e^{j\overline{\phi} _{S,m}} \right\} \mathbb{E}\left\{e^{-j\overline{\phi} _{S,j}}\right\}\nonumber\\
     &=M \Omega_{S} \left( 1  + (M-1)  \frac{\pi}{4} Q \right),
\end{align}
where $\bar \phi_{S,m}=\phi_{S_m} - \hat \phi_{S_m}$ is the quantization noise (error) introduced by $b$-bit phase shifters at the relay's receiving antennas with $\mathbb{E}\left \{ e^{j\overline{\phi} _{S,m}} \right \}=\mathbb{E}\left \{ e^{-j\overline{\phi} _{S,m}} \right \}=\frac{2^b}{\pi}\sin\left ( \frac{\pi}{2^b} \right )=\sqrt{Q}$. Hence, $\gamma_S^{Ra}$ is given by 
\begin{align}\label{ez3.1}
    \gamma_S^{Ra}=\frac{P_S \Omega_{S} \left( 1  + (M-1)  \frac{\pi}{4} Q \right)}{N_0 + \frac{\eta}{N}P_R} . 
\end{align}

On the other hand, the received SINR at the destination, denoted by $\gamma_D^{Ra}$, is given by 
\begin{align}\label{ez4}
    \gamma_D^{Ra}=\frac{P_R\left \{ \left | x_R \right |^2 \right \} \mathbb{E}_{\mathbf{h}_{D}}\left\{\left | \frac{1}{\sqrt{N}}\mathbf{v}^T \mathbf{h}_{D}\right | ^2  \right\}}{\mathbb{E}_{w_D} \left \{  w_D^2  \right \}},
\end{align}
where the only term in \eqref{ez4} different from \eqref{e20} is $\mathbb{E}_{\mathbf{h}_D}\left\{\left | \frac{1}{\sqrt{N}}\mathbf{v}^T \mathbf{h}_D\right | ^2  \right\}$.  After similar derivations to \eqref{ez3}, we  obtain $\gamma_D^{Ra}$ as 
\begin{align}\label{ez5}
    \gamma_D^{Ra}= \frac{P_R\Omega_{D} \left (1 + (N-1)\frac{\pi}{4} Q \right )}{N_0}.
\end{align}
As a result of \eqref{ez3.1} and \eqref{ez5}, we can conclude that $R^{Ra}$ is given by \eqref{q3.00}. 
\end{appendices}
\bibliographystyle{IEEEtran}
\bibliography{citations.bib}

\end{document}